\documentclass{amsart}
\usepackage{graphicx}
\usepackage{xcolor}
\usepackage[normalem]{ulem}
\usepackage[colorinlistoftodos]{todonotes}
\usepackage[utf8]{inputenc}
\usepackage{changepage}
\usepackage{amsmath}
\usepackage{tikz}
\usepackage{amsthm}
\usepackage{amssymb}
\usepackage{enumerate}
\usepackage{xcolor}
\usepackage[ruled,vlined]{algorithm2e}
\usepackage[english]{babel}
\usepackage{ulem}
\usepackage{amsaddr}
\usepackage{lineno}
\usepackage{setspace}
\usepackage{xcolor}
\usepackage{hyperref}

\SetCommentSty{mycommfont}
\SetKwInput{KwInput}{Input} 
\SetKwInput{KwOutput}{Output} 
\usepackage[utf8]{inputenc} 

\usepackage[english]{babel}
 
\newtheorem{theorem}{Theorem}
\newtheorem{lemma}{Lemma}
\newtheorem{corollary}{Corollary}

\theoremstyle{definition}
\newtheorem{definition}{Definition}
\newtheorem{example}{Example}

\begin{document}

\title[Semi-Directed Networks from Quarnets]{Constructing Semi-Directed Level-1 Phylogenetic Networks from Quarnets}


\author{Sophia Huebler} \address{University of Utah}
\author{Rachel Morris}
\address{North Carolina State University}

\author{Joseph Rusinko}
\address{Hobart and William Smith Colleges}

\begin{abstract} 

Semi-directed networks provide a graphical structure for describing the evolutionary history of organisms in the presence of hybridization. We introduce two algorithms for reconstructing semi-directed level-1 phylogenetic networks from their complete set of 4-leaf subnetworks, known as quarnets. The sequential algorithm begins with a single quarnet and adds one leaf at a time until the entire network has been reconstructed. The cherry-blob algorithm is a novel approach inspired by cherry-picking techniques on trees.

\end{abstract}
\maketitle


\section*{Introduction}

Phylogenetics explores evolutionary relationships among species. Historically, these relationships have been depicted using phylogenetic trees, where branching points represent the divergence of new species. However, trees are insufficient to capture the uncertainty in genetic data due to noise or incomplete lineage sorting, nor can they model reticulation events such as hybridization and horizontal gene transfer \cite{BAPTESTE2013439}. As such, research has begun to use phylogenetic networks, which can include one or more cycles, to more accurately model evolution. Given that the class of networks with $n$ leaves is infinite, most research assumes that phylogenetic networks follow specific structural constraints. Many classes of network restrictions have been proposed in this context, including restrictions on the number of reticulation events that can appear on the same connected component, which defines the level of the network. Other common features used for classification are the total number of reticulation events and whether the edges are directed, undirected, or semi-directed. An excellent survey of the implications of these choices can be found in \cite{kong2022classes}.

To study phylogenetic trees or networks, one must first estimate them from genetic data. However, directly estimating a tree or network can be resource-intensive when there are large numbers of taxa. One approach to addressing this challenge is to infer evolutionary history on subsets of the data and glue this information together to understand the evolutionary history of all species of interest. Subsets of four taxa are frequently preferred as base units. As a result \emph{quartets}, or 4-leaf trees, have been widely used as building blocks for phylogenetic tree reconstruction \cite{casanellas2023designing,chifman2014quartet,larget2010bucky, astral,rusinko2012invariant,strimmer1996quartet}. Phylogenetic tree reconstruction typically involves two steps, with algorithms often focusing on one of the two tasks: first, inferring relationships among four-taxa subsets from genetic data and then constructing a tree or network that reflects these relationships. Leveraging existing methods to determine quartets from genetic data, various algorithms to reconstruct phylogenetic networks also choose to use quartets as the building blocks \cite{allman2019nanuq, Allman2024_NANUQ+, gambette2012quartets,keijsper2014reconstructing,solis-lemusSnaQ}. These algorithms range from the SNaQ pseudolikelihood method that reconstructs a semi-directed network \cite{solis-lemusSnaQ} to NANUQ and NANUQ+ which infer the network topology and network, respectively, by estimating the distance between pairs of taxa \cite{allman2019nanuq, Allman2024_NANUQ+}. An additional reconstruction algorithm uses quartets as an intermediary in estimating networks from SNP data \cite{warnow_statistically_2024}.

While many phylogenetic networks can be inferred by quartets, networks that display complex relationships between taxa as semi-directed cycles cannot be  uniquely determined from their displayed quartets \cite{banosIdentifying}. 
Using unrooted four-leaf networks called \emph{quarnets} that allow for semi-directed cycles yields stronger theoretical guarantees in reconstructing phylogenetic networks. In particular, recent work demonstrates that all level-2 semi-directed networks are encoded by their quarnets \cite{huber2025sufficientreconstruct}. Various works also provide means to infer quarnets from DNA sequences, making network reconstruction from quarnets possible. For the Juke-Cantor model, statistical learning via SVMs has been successful at learning complete level-1 quarnets \cite{barton2022statlearning}. Another method uses algebraic invariants to infer quarnets and to identify the network topology \cite{MartinAlgInvariants}.

Related to reconstruction algorithms is the concept of network identifiability, which refers to whether an evolution history can be inferred from phylogenetic models. For the Jukes-Cantor model, one can identify the topology of the 4-leaf network and the orientation of some, but not all, edges \cite{gross}. Identifiability in the case of a single reticulation event on a cycle of at least length 4 has been demonstrated for the K2P and K3P models \cite{HolleringSullivant}. Follow-up work proved the identifiability of triangle-free level-1 phylogenetic networks for all three previously mentioned models \cite{GrossDistinguishingL1}.

In this article, we propose two approaches for reconstructing semi-directed network $N$ by piecing together the complete collection of 4-leaf subnetworks of $N$. The \emph{sequential algorithm} iteratively adds leaves to a base network, serving as a generalization of quartet-puzzling \cite{strimmer1996quartet}. The second method, which we refer to as the \emph{cherry-blob algorithm}, generalizes cherry-picking on trees \cite{cherrypicking2} by recursively identifying external structures of a network. Both methods run in polynomial time in $n$, the number of leaves. For an alternative approach to reconstructing networks from subsets of displayed quarnets, see \cite{frohn_reconstructing_2025}. A fourth algorithm, Squirrel, provides an additional method to reconstruct triangle-free level-1 networks for a full set of quarnets \cite{HoltgrefeSQUIRREL}.

The structure of this article is as follows: Section~\ref{sec:Background} introduces the necessary definitions and notation, Section \ref{sequential} presents the sequential algorithm, and finally, Section~\ref{Cherry-Blob} describes the cherry-blob algorithm.

\section{Background}
\label{sec:Background}

\subsection{Semi-directed, binary level-1 phylogenetic networks}

\label{networkDef}
We adopt the network language and notation from \cite{gross}. A network $N=(V, E)$ consists of a collection of vertices $V$ and unordered pairs of vertices that form the edge set $E$. The \emph{degree} of a vertex is given by the number of edges incident to the vertex. Any vertex that has degree one is a \emph{leaf} while any vertex of higher degree is an \emph{internal vertex}. In a \emph{binary} network, all internal vertices are of degree three, except potentially the \emph{root} $\rho$ which has degree two. The set of leaves or support of a network $N$ is denoted equivalently by $supp(N)$ or $X$ if the network is clear from context. A network can be \emph{directed} where all edges are given an arrow representing a direction of flow, \emph{semi-directed} where only certain edges are assigned a direction, or \emph{undirected}. In directed and semi-directed networks, vertices may have an in-degree given by the number of incident inward-flowing edges and an out-degree given by the number of incident outward-flowing edges.

A \emph{phylogenetic network} is a rooted, binary, acyclic, directed graph that satisfies the following properties:  parallel edges are suppressed,  the root has out-degree 2,  all leaves have in-degree 1, and all internal vertices have either in-degree 1 and out-degree 2 or out-degree 1 and in-degree 2. In a phylogenetic network, a reticulation event is represented by a \emph{reticulation vertex}, which has in-degree two and out-degree one. Reticulation events appear in cycles in the underlying topology of the network. The two edges directed into a reticulation vertex are called \emph{reticulation edges}, and the edge directed away from the reticulation vertex is referred to as an \emph{out edge}. A leaf connected to an out edge is called an \emph{out leaf}. A \emph{semi-directed phylogenetic network} is obtained from a phylogenetic network by undirecting all edges except for the reticulation edges and suppressing all degree two vertices and parallel edges. 

A graph is \emph{connected} if there exists a path between any two vertices. A graph is \emph{biconnected} if deleting any vertex results in a graph that remains connected. A \emph{biconnected component} is a maximal subgraph that is biconnected. A network is \emph{level-$k$} if, in every biconnected component, there are at most $k$ reticulation vertices. We will be considering level-1 networks, which include the set of phylogenetic trees as well as networks whose cycles are disjoint. Henceforth, unless otherwise specified, the term \emph{network} will refer to an unrooted, binary, level-1, semi-directed phylogenetic network.

\subsection{Restrictions to Quarnets}

Our reconstruction algorithms identify a network's features from network restrictions to a smaller leaf set, following the strategy outlined by Gross and Long \cite{gross}, which we summarize as the Network Restriction Algorithm.

\begin{algorithm}[h]
\label{restrictions}
\DontPrintSemicolon
 \KwInput{An unrooted semi-directed network $N$, with leaf set $X$, let $Y \subseteq X$.}
 \KwOutput{The restriction of $N$ to leaf set $Y$, $N|_{Y}$.}
 Place a root $\rho$ on a valid root location. (See \cite{validroot} for an algorithm.)
 
 Direct all undirected edges away from $\rho$.
 
 Find the union of all directed paths from $\rho$ to all $y\in Y$.

 Remove all vertices and edges not on any of the directed paths.

 \While{Restriction has non-root degree 2 vertices or parallel edges}{
 Suppress all degree 2 vertices except for $\rho$.

 Remove all parallel edges.
 }
 
 Suppress $\rho$.
 
 Undirect all edges that are not reticulation edges.

\Return{Semi-directed network $N|_{Y}$ }
\caption{Network Restriction}
\end{algorithm}

In general, a network $N$ \emph{displays} another network $N'$ if $N|_{supp(N')}=N'$. If we consider a set of restricted networks, a network that displays all the restrictions is called a \emph{parent network}. In general, a collection of restricted networks may have multiple parent networks. However, the results of \cite{huber2025sufficientreconstruct} can be interpreted as follows: the set of all displayed quarnets of a semi-directed level-1 network has a unique parent.

\begin{example}
     We consider a semi-directed phylogenetic network and six of its displayed quarnets in Figure~\ref{fig:quarnets}. The network displays 70 quarnets but only six semi-directed quarnet topologies up to relabeling the leaves. Every quarnet restriction of any semi-directed phylogenetic network shares the same semi-directed topology as one of these six quarnets. We refer to these topologies as quartet trees, single triangles, squares, and double triangles. 
\end{example}

\begin{figure}[ht]
 \centering
 \includegraphics[width = 340pt]{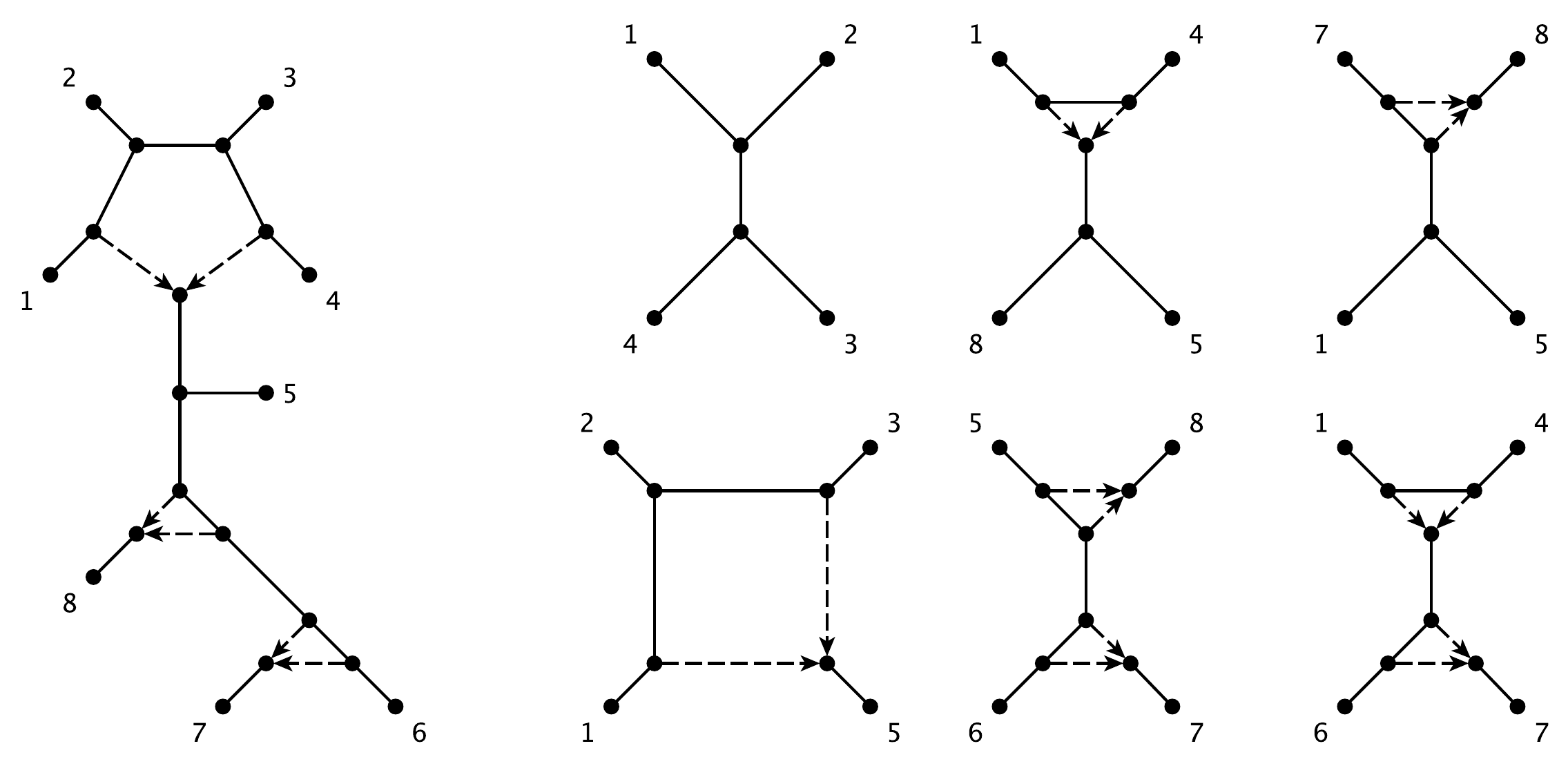}
 \caption{ A phylogenetic network and six of the quarnets it displays.} 
 \label{fig:quarnets}
\end{figure}

We denote a set of quarnets by $\mathcal{Q}$. If quarnets in a set can appear as restrictions of a shared parent network, we say $\mathcal{Q}$ is a \emph{compatible} set of quarnets. The set of all quarnets displayed by a parent network $N$ is called the \emph{complete quarnet set}, denoted by $\mathcal{Q}(N)$. If we want to consider all quarnets in $\mathcal Q$ that contain the leaf $x\in X$, we denote this subset by $\mathcal Q_x$. 
The two main algorithms in this article aim to reverse the restriction process by reconstructing a network from a complete quarnet set.

\subsection{Local Network Components}
Our proposed algorithms for reconstructing networks from quarnets depend on identifying local network components. Biconnected components of a network are referred to as \emph{semi-directed cycles}. A \emph{blob} typically refers to a biconnected component with at least 3 vertices \cite{validroot}. For this article, we further extend the definition to include leaves connected by a single trivial cut edge such that a blob is redefined as the subgraph consisting of a semi-directed cycle and all leaves adjacent to vertices on it. The term \emph{cycle} here refers to the underlying undirected cycle. A vertex $v$ is said to be \emph{on a cycle} if the vertex sequence of the undirected cycle includes $v$, and similarly, an edge $e$ is \emph{on a cycle} if the edge sequence of the undirected cycle includes $e$.

A \emph{cut edge} is an edge in $N$ whose removal disconnects $N$ into two components. The cut edge necessarily partitions the leaves into two sets. A cut edge is \emph{trivial} if one of these sets contains only a single leaf. If removing a cut edge disconnects a subgraph consisting of two leaves adjacent to a single vertex, then this two-leaf subgraph is called a \emph{tree cherry}. The quartet tree has two tree cherries, and the single triangles have one tree cherry. If a blob is incident to exactly one non-trivial cut edge, it is called an \emph{exterior blob}. In Figure~\ref{fig:quarnets}, the network exhibits two exterior blobs, one the cycle associated with the leaves 1, 2, 3, and 4 and the other the cycle associated with the leaves 6 and 7. We refer to exterior blobs of cycle length three as \emph{reticulation cherries}. Both single triangles have one reticulation cherry, whereas both double triangles have two reticulation cherries. In this paper, the term \emph{cherry} refers to both tree and reticulation cherries. An \emph{exterior structure} is any cherry or exterior blob. A network without exterior structures is called a \emph{sunlet}, where an $n$-sunlet has $n$ leaves and one semi-directed cycle of length $n$. For instance, the square quarnet is a 4-sunlet.

\begin{definition}
\label{def:PendantPair}
Two leaves $x,y$ of a network $N$ form a \emph{pendant pair} if there exists an edge $E=\{u,v\}$ such that $x$ is adjacent to $u$ and $y$ is adjacent to $v$.
\end{definition}
While quartet trees and reticulation cherries exhibit pendant pairs, we will use this term primarily as a tool for ordering the leaves of a sunlet.

Having established the fundamental definitions and terminology for semi-directed phylogenetic networks, we now turn to developing algorithms for reconstructing such networks. These algorithms leverage the properties of quarnets and the structural features of level-1 networks to iteratively assemble parent networks from their subnetworks.

\section{Sequential Algorithm}
\label{sequential}
This section introduces the \emph{sequential algorithm} for constructing an $n$-leaf network from its complete quarnet set. The algorithm begins with a single quarnet and iteratively attaches one leaf at a time, determining the attachment location through a voting procedure. Notably, the resulting network is independent of the choice of initial quarnet and the ordering of the leaf additions.

\subsection{Leaf Attachments Determine Networks}
\label{leaf}
In the sequential algorithm, we seek to construct level-1 phylogenetic networks by iteratively adding leaves to the network. We define the following three methods to attach a leaf, and these attachments are depicted in Figure~\ref{fig:attachments}. By definition, these single leaf attachments do not produce networks that are level-2 or higher.  

\begin{figure}
 \centering
 \includegraphics[width=10cm]{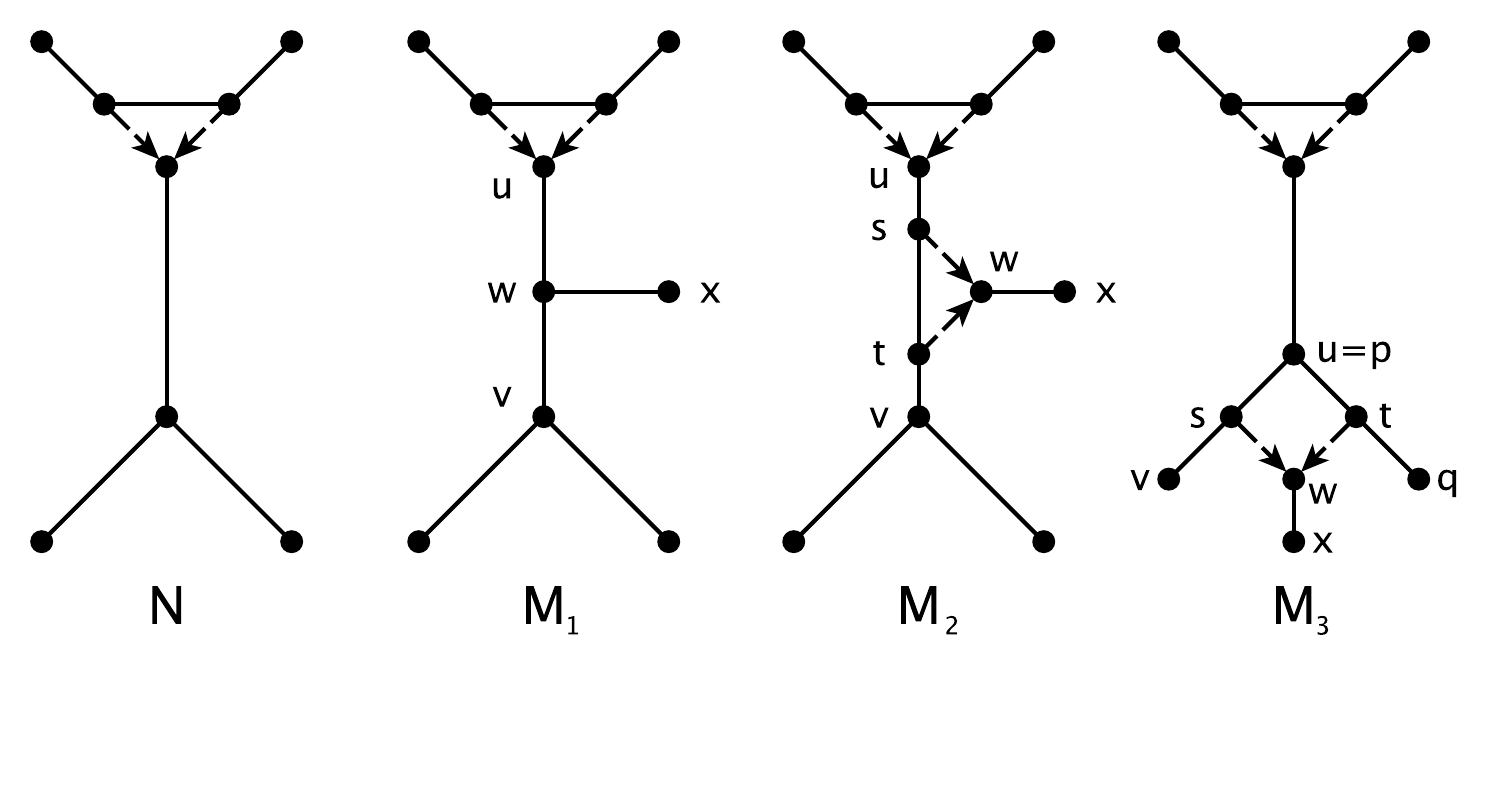}
 \caption{The three single leaf attachment functions ($M_1$, $M_2$, and $M_3$) applied to a network $N$.}
 \label{fig:attachments}
\end{figure}

\begin{definition}
\label{singleleaf}
Given a network $N$ and a leaf $x \not\in supp(N)$, we define the \emph{single leaf attachment functions} $M_1,M_2,$ and $M_3$ as follows:
\begin{enumerate}[\null]
\item \textbf{Insert a leaf edge $(M_1)$}: 
Given an edge $e = \{u,v\} \in N$, let $N'=M_1(N, x, e)$ be the network formed by subdividing $\{u,v\}$ into  $\{u,w\}$ and $\{w,v\}$ and introducing $x$ as leaf of $N'$ by constructing the additional edge $\{x,w\}$. The edges $\{u,w\}$, $\{w,v\}$, and $\{x,w\}$ are undirected in $N'$ unless $e$ was a reticulation edge on $N$.  Without loss of generality, if $e$  was directed from $u$ towards $v$, then $\{w, v\}$ is a reticulation edge toward $v$ in $N'$ and both $\{u,w\}$ and $\{w,x\}$ are undirected.

\item \textbf{Create a 3-cycle $(M_2)$}: Given an edge $e = \{u,v\}$ not on a cycle (and hence undirected), the network $N' = M_2(N, x, e, r)$ is formed by inserting vertices $s, t, w, x$, replacing $e$ with the edges $\{u,s\}, \{s,t\}, \{t,v\}$, and adding edges $\{s,w\}, \{t,w\}, \{w,x\}$ to create a 3-cycle. The reticulation vertex $r \in \{s,t,w\}$ determines the directed reticulation edges.
 
\item \textbf{Create a larger cycle $(M_3)$}: For any two edges $e_1 = \{u,v\}$ and $e_2= \{p,q\}$, let $P(e_1,e_2)$ be the  edges contained in any  path between $e_1$ and $e_2$. Note that $e_1$ and $e_2$ may share a single vertex. Let $\{e_1,e_2\}$  be a pair of edge such that $P(e_1,e_2)\cap E(C)=\emptyset$ for all cycles $C$ of $N$. The network $N' = M_3(N, x, e_1, e_2)$ is constructed by introducing vertices $s, t, w, x$, replacing $e_1$ with edges $\{u,s\}$ and $\{s,v\}$, and $e_2$ with $\{p,t\}$ and $\{t,q\}$. The new directed edges $\{s,w\}, \{t,w\}$ lead to the reticulation vertex $w$, with a new outgoing edge $\{w,x\}$. 
\end{enumerate}
\end{definition}

A network $N$ can be \emph{realized} by a sequence of leaf attachments if it can be built starting from a base quarnet and applying a series of leaf attachments.
We now show that every level-1 network can be constructed by a series of leaf attachments by an inductive argument. 

\begin{theorem}
\label{SProof}
Let $(x_1,x_2,\cdots, x_n)$ be an ordering of the leaves of a level-1 network $N$ with at least four leaves. Then $N$ can be realized by a sequence of single leaf attachments to the base quarnet $N|_{\{x_1,x_2,x_3,x_4\}}$. 
\end{theorem}

\begin{proof}
We proceed by induction on the number of leaves of a level-1 network $N$. Let $(x_1,x_2,\cdots, x_n)$ be an ordering of the leaves $N$. No attachments are needed for networks with four leaves since the network is a quarnet. Assume that for any sequence  $(x_1,x_2,\cdots, x_n)$ of its leaves, an $n$-leaf network $N$ can be realized as a sequence of single leaf attachments to $N|_{\{x_1,x_2,x_3,x_4\}}$. Now consider an arbitrary sequence $(x'_1,x'_2,\cdots, x'_n,x'_{n+1})$ of an $(n+1)$-leaf network $N'$. 

Then, by our inductive hypothesis, $N=N'|_{\{x'_1,x'_2,\cdots, x_n'\}}$ can be realized as a sequence of leaf attachments from the quarnet $q=N'|_{\{x'_1,x'_2,x'_3,x'_4\}}$.
Now consider the leaf $x=x'_{n+1}$. Let $w$ be the vertex adjacent to $x$ in $N'$. 

First, suppose $w$ is not on a 3-cycle and $x$ is not an out leaf. Let $u,v$ be the other two vertices adjacent to $w$ in addition to $x$ in $N'$. From $N'$, we can determine $N$ by removing the vertex $x$ and the incident edge $\{x,w\}$. Now, $w$ is a degree 2 vertex, so we suppress it. If either $\{u,w\}$ or $\{v,w\}$ was directed in $N'$, let $e= \{u,w\}$ have the same direction in $N$.  By the inductive hypothesis, $N$ can be realized by a sequence of leaf attachments given the ordering $(x'_1,x'_2,\cdots, x_n')$. The claim follows by noting that $N' = M_1(N, x, e)$. 

Next, suppose $x$ is on a 3-cycle. Let $s$ and $t$ be the other two vertices on the 3-cycle. Let $r\in \{s,t,w\}$ be the reticulation vertex on the cycle. Furthermore, let $u$ be the third vertex adjacent to $s$ and $v$ be the third vertex adjacent to $t$. To determine the $n$-leaf network $N$ from $N'$, remove $x$ and the edge $\{x,w\}$. Undirect the reticulation edges on the cycle between $s,t, w$ as removing $x$ collapses the cycle into a single undirected edge. To this end, suppress all degree two vertices and parallel edges until only the edge $e = \{u,v\}$ remains. By the inductive hypothesis, $N$ can be realized by a sequence of leaf attachments given the ordering $(x'_1,x'_2,\cdots, x_n')$. The claim follows by noting that $N' = M_2(N, x, e, r)$. 

Finally, if $w$ is a reticulation vertex on a $k$-cycle for an integer $k > 3$, then $x$ is an out leaf. Let $s$ and $t$ be the two vertices adjacent to $w$ on the $k$-cycle. Let $u$ and $v$ be the two other vertices adjacent to $s$ and $p$ and $q$ be the two other vertices adjacent to $t$. To determine $N$, remove vertices $x$ and $w$, the edge $\{x,w\}$, and the reticulation edges $\{s,w\}$ and $\{t,w\}$. Then, $s$ and $t$ are degree two vertices, so suppress them. Let $e_1 = \{u,v\}$ and $e_2 = \{p, q\}$. By the inductive hypothesis, $N$ can be realized by a sequence of leaf attachments given the ordering $(x'_1,x'_2,\cdots, x_n')$. The claim follows by noting that $N'=M_3(N,x, e_1, e_2)$.

\end{proof}

\subsection{Sequential Algorithm}
\label{algorithm1}

By Theorem \ref{SProof}, we can realize any network as a sequence of single leaf attachments to a base quartet. However, the construction required knowing the network in advance. We will show that the sequence of leaf attachments can also be extracted directly from the complete set of quarnets. The process involves selecting one quarnet as a base and sequentially attaching the remaining leaves to the network. To begin, we rigorously define when a leaf attachment is allowed and when it is optimal.

\begin{definition}
\label{voting}
Let $N$ be a network with leaf set $X$. Let $\mathcal{Q}$ be a set of quarnets such that $\mathcal{Q}(N) \subset \mathcal{Q}$ and let $x$ be a leaf such that $x \in \bigcup_{q\in \mathcal Q} supp(q) \setminus X$. A single leaf attachment $M(N,x)$ is \emph{allowed} on $N$ by $q \in \mathcal{Q}_x$ if $N'=M(N,x)$ displays $q$. A single leaf attachment $M(N,x)$ is an \emph{optimal attachment} if $M(N,x)$ is allowed by all $q \in \mathcal{Q}_x$.
\end{definition}

In order to reconstruct a network from its quarnets, we must be able to determine the optimal attachment of a new leaf at each step. The following theorem demonstrates that such an optimal attachment for any leaf exists and is unique.

\begin{theorem}
\label{thm:optimal}
Let $N'$ be a network with leaf set $X'$. Let $X \subset X'$ with $|X|\ge 4$, and let $N=N'|_X$. If $x \in X' \setminus X$, there is a unique optimal attachment of $x$ to $N$ with respect to $\mathcal{Q}(N')$.
\end{theorem}

\begin{proof}
Let $N'$ be a network with leaf set $X'$. Let $X \subset X'$ with $|X|\ge 4$, and let $N=N'|_X$. If $x \in X' \setminus X$. Let $N^*$ be the restriction of $N'$ to the set $X \cup \{x\}$. It follows that $\mathcal{Q}({N^*}) = \mathcal{Q}(N'|_{X \cup \{x\}})$. 

We choose an ordering of the leaves of $X \cup \{x\}$ such that $x$ is the final leaf. 
By Theorem~\ref{SProof}, $N^*$ can be realized by a sequence of single leaf attachments to a quarnet base according to this ordering. Thus, $N^*=M(N,x)$ for some single leaf attachment. Thus, $=M(N,x)$ is an optimal attachment. The uniqueness follows from $N^*$ being the unique restriction of $N'$ to the set $X \cup \{x\}$, and the fact that distinct single leaf attachments must result in non-isomorphic networks. 
\end{proof}

We now give our first algorithm for reconstructing the networks from quarnets:

\begin{algorithm}[h]
\label{Salg}
\DontPrintSemicolon
 
\KwInput{A complete set of quarnets $\mathcal{Q}(N)$ on $n$ leaves}
\KwOutput{A unique $n$-leaf network $Sequential(\mathcal{Q}(N)) = N$}

Assign an arbitrary ordering to the set of leaves $X$ on $\mathcal{Q}(N)$ such that $X = \{x_1, \dots, x_n\}$

Let $N_4$ be the quarnet with support $\{x_1, x_2, x_3, x_4\}$

\For{$i \gets 4$ \KwTo $n-1$}{
    Identify the unique optimal single leaf attachment, $M(N_i,x_{i+1})$, by checking if $\mathcal{Q}(M(N_i,x_{i+1}))=\mathcal{Q}(N|_{\{x_1, x_2,\cdots x_{i+1}\}})$ for at most $4|E(N_i)|+ \binom{|E(N_i)|}{2}$ single leaf attachments.
    
    Set $N_{i+1}= M(N_i,x_{i+1})$ 
}
\Return{$N_n$}
\caption{Sequential Algorithm}
\end{algorithm}

\begin{theorem}
Let $\mathcal{Q}(N)$ be the complete set of quarnets of a level-1 network $N$. Then, $Sequential(\mathcal{Q}(N)) = N$.
\end{theorem}

\begin{proof}
Let $f(\mathcal{Q}(N))$ denote the network output from the sequential algorithm with input set $\mathcal{Q}(N)$. If $N$ has four leaves, then $f(\mathcal{Q}(N))$ will return the base quarnet, which must be the network $N$. If $N$ is has more than $n$ taxa for $n>4$, then we have by definition $N_4 = N|_{\{x_1,x_2,x_3,x_4\}}$. For $4 \le i \le n-1$, the proof of Theorem~\ref{thm:optimal} ensures that $\mathcal{Q}(N_{i+1})=\mathcal Q(N|_{\{x_1,\cdots,x_{i+1}\}})$ hence $N_{i+1}=N|_{\{x_1,\cdots,x_{i+1}\}}$. In particular $f(\mathcal{Q}(N))=N_n=N$ as claimed.
\end{proof}

\begin{theorem}
The sequential algorithm with input set $\mathcal{Q}(N)$ for an $n$-leaf network has worst-case time complexity $O(n^6)$.
\end{theorem}

\begin{proof}
We analyze the cost of the algorithm as a function of the number of leaves $n$ of the network $N$ by considering each phase of its execution.

\textbf{Initial Setup.}
The algorithm begins by setting $\mathcal{Q} = \mathcal{Q}(N)$, which we assume is provided as part of the input. Accessing or copying this data takes $O(n^4)$ time.

\textbf{Sequential addition.}
At each iteration $i$ from $4$ to $n-1$, the algorithm builds a new network $N_{i+1}$ by attaching leaf $x_{i+1}$ to the current network $N_i$. For each such step:

\begin{itemize}
    \item The number of candidate single-leaf attachments is at most $O(i^2)$, since for a level-1 network on $i$ leaves, the number of edges $|E(N_i)| = O(i)$, and the algorithm considers up to $4|E(N_i)| + \binom{|E(N_i)|}{2} = O(i^2)$ candidate attachments.
    \item For each candidate network $M(N_i, x_{i+1})$, the algorithm computes the set of quarnets involving $x_{i+1}$. These quarnets are supported on subsets of four leaves that include $x_{i+1}$ and any three of the $i$ existing leaves, so their number is $\binom{i}{3} = O(i^3)$.
    \item Each quarnet comparison with the corresponding quarnet from $\mathcal{Q}(N|_{\{x_1, \dots, x_{i+1}\}})$ is assumed to take $O(1)$ time.
\end{itemize}

Thus, the total cost per iteration $i$ is:
\[
O(i^2) \text{ candidate attachments} \times O(i^3) \text{ quarnet checks} = O(i^5)= O(n^5)
\]

Since there are $n-4$ required additions the run time of the sequential addition stages is $ n \times O(n^5)=n^6$

\textbf{Total run time.}
The sequential algorithm runs in \[O(n^4) \text{ (setup) }  + O(n^6) \text{ (sequential addition) } = O(n^6).\]
\end{proof}

\subsection{Example of the Sequential Algorithm}
We now work through an example of using the sequential algorithm to construct an 8-leaf network. The complete set of quarnets for an 8-leaf network consists of 70 quarnets. Since not all quarnets are presented explicitly, we reference only the quarnets shown in Figure~\ref{fig:seqqnets}. Interestingly, reconstructing the network $N$ will only require these six quarnets, which we will refer to as $q_1, q_2,\cdots q_6$ throughout this example. 

\begin{figure}
 \centering
 \includegraphics[height=2.75cm]{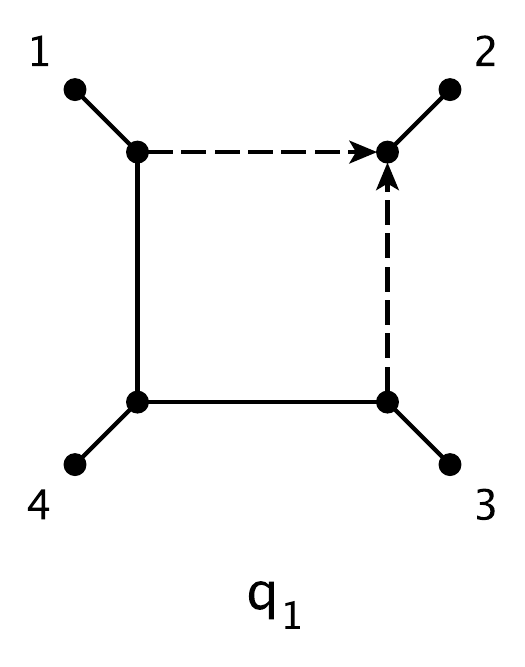}
 \includegraphics[height=2.75cm]{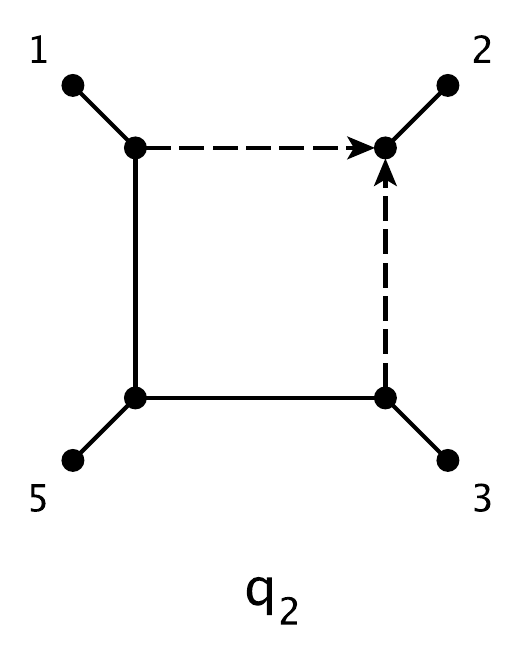}
 \includegraphics[height=2.75cm]{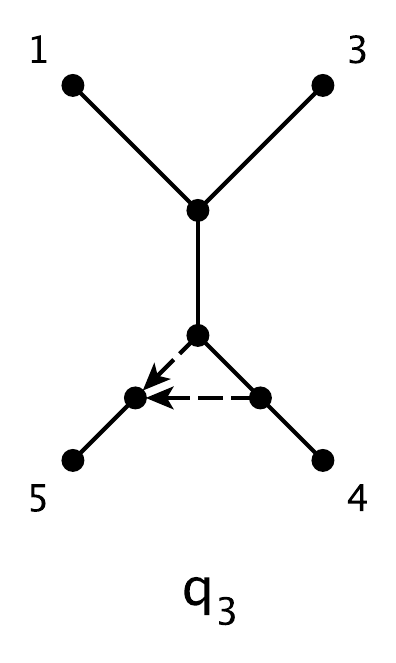}
 \includegraphics[height=2.75cm]{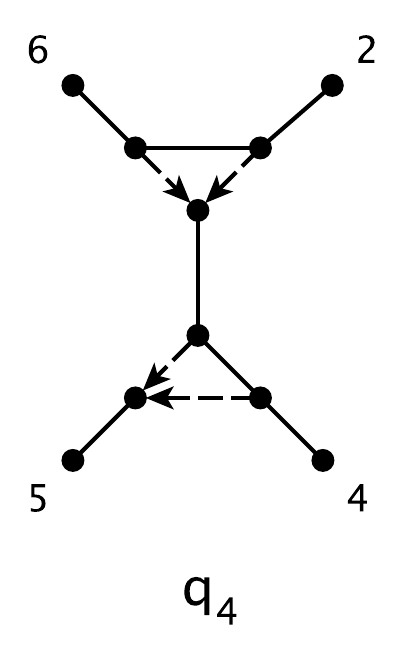}
 \includegraphics[height=2.75cm]{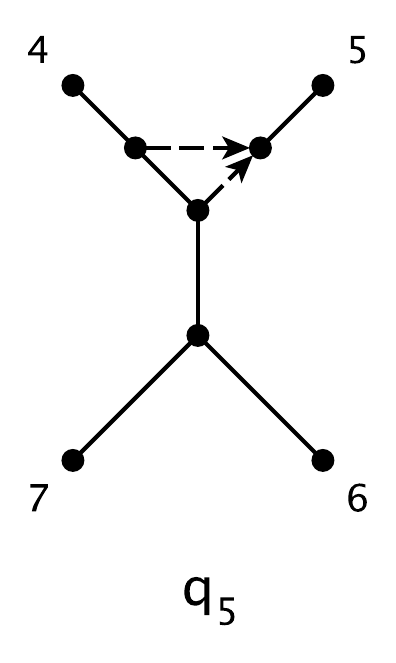}
 \includegraphics[height=2.75cm]{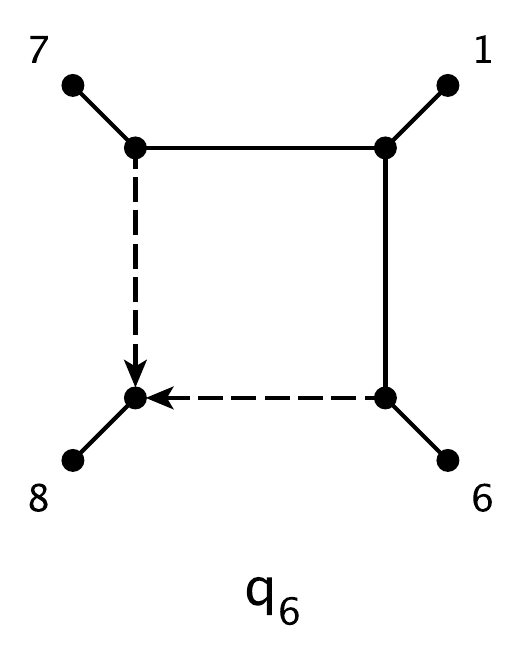}
 \caption{A subset of the complete quarnet set $\mathcal{Q}(N)$  for an unknown 8 leaf phylogenetic network $N$.}
 \label{fig:seqqnets}
\end{figure}

Let the leaf ordering be numerical (i.e., $1, 2, \dots, 8$). We use $q_1$ as the base quarnet or $N_4$. Any single leaf attachment to $N_4$ of type $M_3$ involves placing a leaf between a pair of edges, which would increase the level. Thus, only attachments $M_1$ and $M_2$ need to be considered when attaching the fifth leaf.

We must identify the unique optimal leaf attachment for leaf 5. In order to do so, we consider $\mathcal Q_5$, the subset of all quarnets containing the leaf 5. 
We begin with the quarnet $q_2$ with $supp(q_2) = \{1, 2, 3, 5\}$, the second quarnet in Figure~\ref{fig:seqqnets}. First, we look at all possible $M_1$ attachments. Based on the order of vertices defined by their adjacency, the placement of leaf 5 on edges $f, g$, or $h$ in Figure~\ref{fig:n4} is allowed. As for $M_2$, the only allowed attachment is on edge $g$, though any orientation is possible. 

\begin{figure}
 \centering
 \includegraphics[height=3cm]{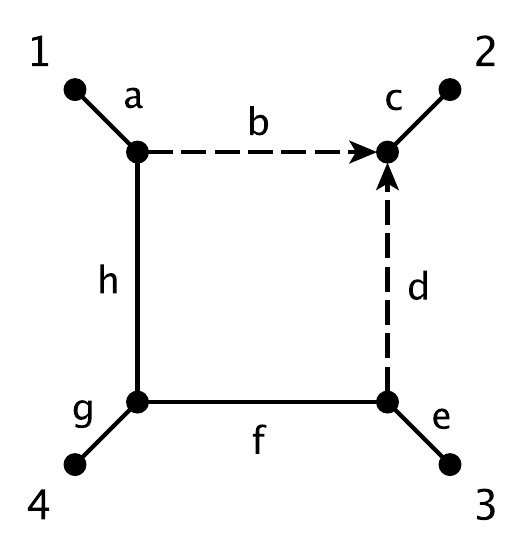} \hspace{1cm} \includegraphics[height=3cm]{InvolveSubmission/Images/SEq2.pdf} \includegraphics[height=3cm]{InvolveSubmission/Images/SEq1.pdf}
 \caption{Base quarnet $N_4$ and two quarnets in $\mathcal{Q}_5(N)$
 for which there is a unique allowable single leaf attachment of leaf $5$.}
  \label{fig:n4}
\end{figure}

Since $q_2$ allows multiple attachments, we examine another quarnet $q_3 \in \mathcal Q_5$. This quarnet disallows $M_1$ on any edge of $N_4$ because performing an $M_1$ attachment would form a square quarnet on the vertices $\{1, 3, 4, 5\}$ which is incompatible with the actual topology of these vertices according to $q_3$. Therefore, only $M_2(N_4, 5, g, w)$ is allowed. Since this is the only allowed attachment, it must be optimal. Thus,  $N_5=M_2(N_4, 5, g, w)$. 

\begin{figure}
 \centering
 \includegraphics[height=3.5cm]{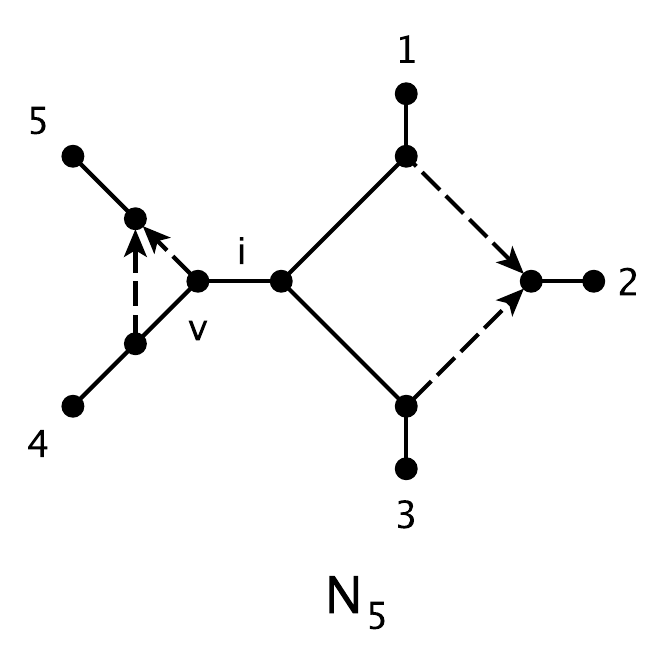} \hfill
 \includegraphics[height=3.5cm]{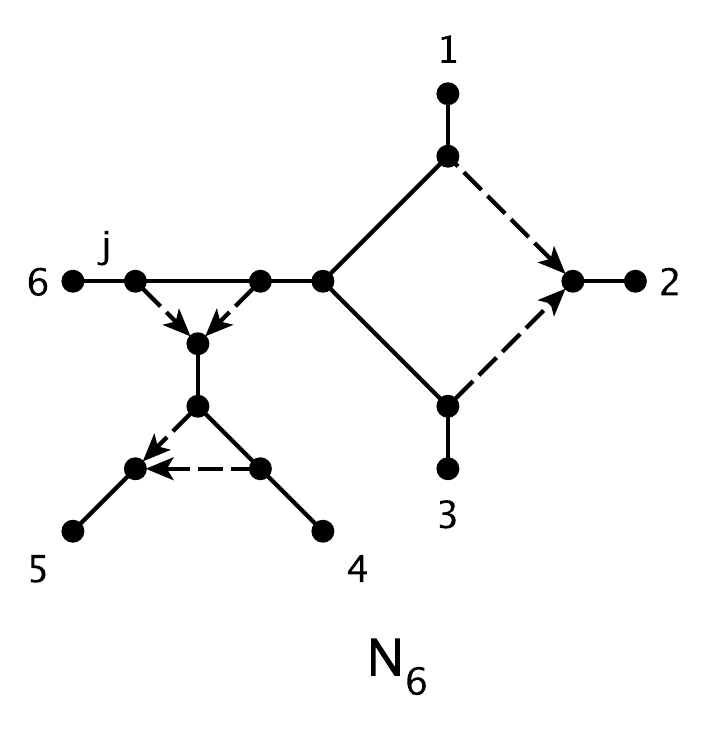} \hfill
 \includegraphics[height=3.5cm]{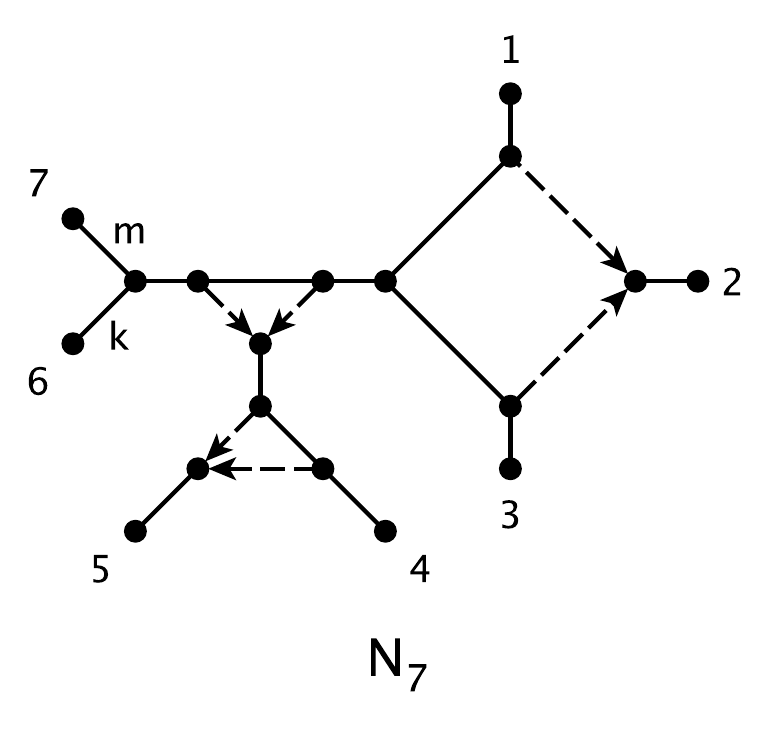}
 \includegraphics[height=3.5cm]{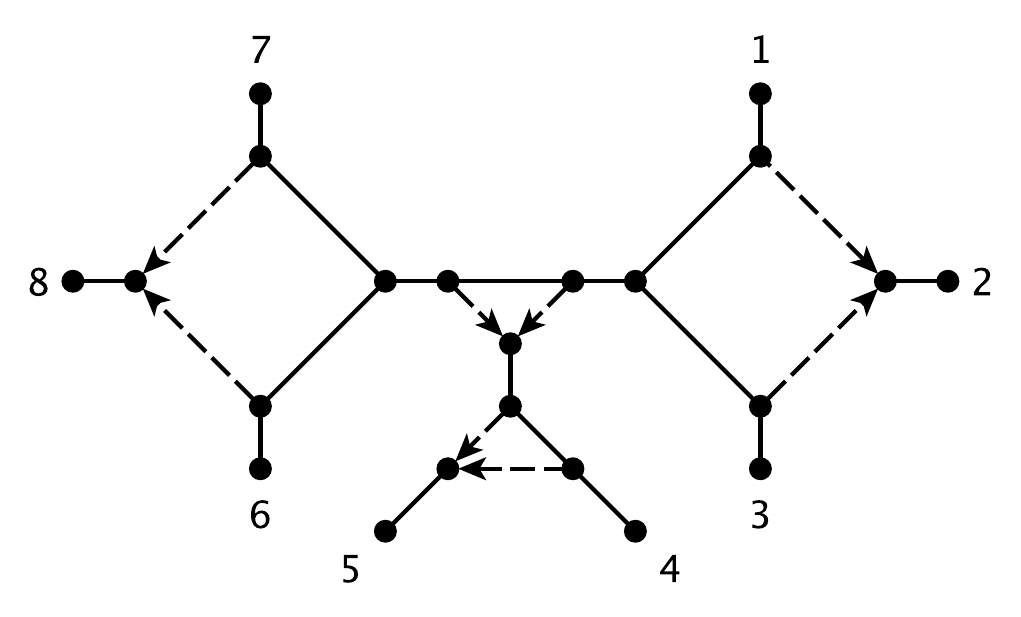}
 \caption{Intermediate networks $N_5$, $N_6$, and $N_7$ determined by the quarnets in Figure~\ref{fig:seqqnets} and the fully reconstructed network $N_8$.}
 \label{fig:intermediates}
\end{figure}

For the remaining leaves, we select quarnets that allow only one attachment, as shown in Figure~\ref{fig:intermediates}. To determine $N_6$, note that $q_4$ allows only $M_2(N_5, 6, i, v)$ because leaf 6 cannot be placed on the same cycle as leaves 2, 4, or 5. To determine $N_7$, we note $q_5$ allows only $M_1(N_6, 7, j)$, as no new cycle is needed. Finally for $N_8$, we observe $q_6$ allows only $M_2(N_7, 8, k, w)$.

Thus, by applying the sequential algorithm to the quarnets shown, we construct an 8-leaf level-1 network.

\section{Cherry-Blob Algorithm}
\label{Cherry-Blob}

\subsection{Cherry-Blob Algorithm Overview}
\label{cbOverview}
Phylogenetic trees can be constructed using cherry-picking methods to identify tree cherries from quartets and recursively build trees \cite{cherrypicking2}. Similarly, we show that networks can be constructed by identifying exterior structures from quarnets. We propose the \emph{cherry-blob algorithm} for reconstructing networks from the complete set of quarnets.

\subsection{Properties of Exterior Structures}
\label{propAS}
 An exterior structure is a locally connected component of a network that contains only trivial cut edges. If we call the support of this component $S$, then the exterior structure $A_S$ is a tree cherry, a reticulation cherry, or an exterior blob (Figure ~\ref{fig:exteriorstructures}). If $A_S$ is a tree cherry, then both leaves in $S$ are adjacent to a common vertex. If $A_S$ is a reticulation cherry or an exterior blob, then all leaves in $S$ are adjacent to vertices on a common semi-directed cycle.

\begin{figure}
 \centering
 \includegraphics[width = 250pt]{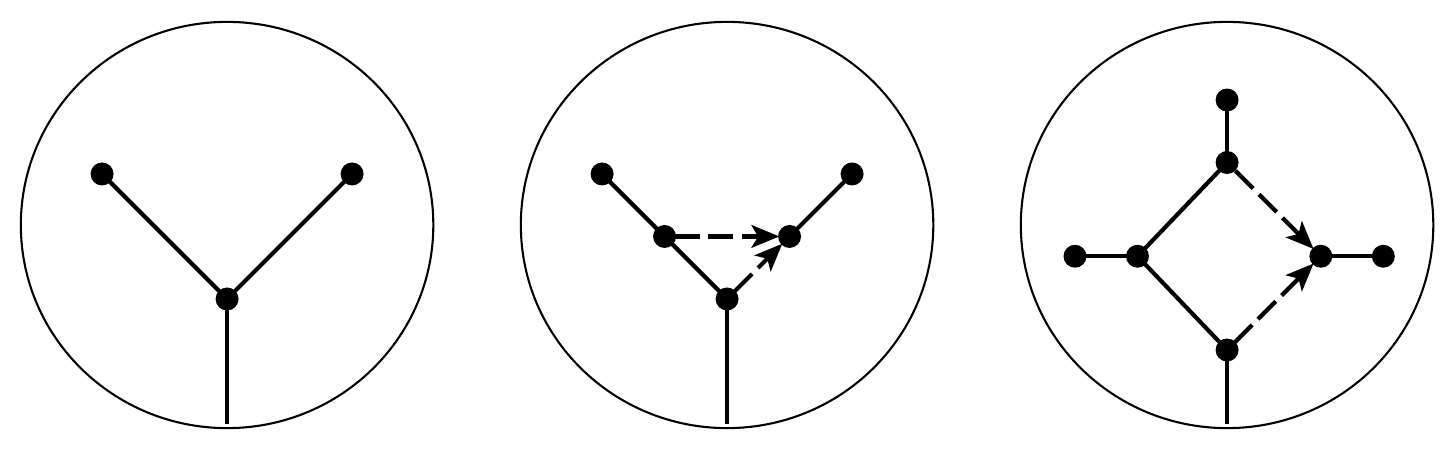}
 \caption{The three types of exterior structures found on networks: (left to right) a tree cherry, a reticulation cherry, and an exterior blob.}
 \label{fig:exteriorstructures}
\end{figure}

The \emph{stem} of an exterior structure $A_S$ is the non-trivial cut edge $\{u,v\}$ which connects $A_S$ to a larger network. Call $v$ the \emph{stem tip} if $v$ is adjacent to the leaves of the cherry or if $v$ is a vertex on the semi-directed cycle in the case of a reticulation cherry or an exterior blob. The other vertex, $u$, is the \emph{stem base}. The subgraph $A_S$ contains the stem tip and every path between the stem tip and a leaf in $S$. The stem tip is the only degree 2 vertex on $A_S$. The stems of these exterior structures are the markers for cutting and gluing networks together.

\subsection{Manipulations of Exterior Structures}
A network $N$ can be modified by cutting or inserting exterior structures by modifying stems, which preserves the level-1 structure of the network. The two operations are denoted by $C$ for cutting an exterior structure and $I$ for inserting an exterior structure. 

\begin{definition}
\label{Cut}
Let $A_S$ be an exterior structure of a network $N$ with stem base $u$ and stem tip $v$. Then $C(N,A_S)$ is the network constructed by removing $A_S$ from $N$ and inserting a placeholder leaf $a_S$ along a new edge $\{u,a_S\}$. This is the \textit{cut procedure}.
\end{definition}

\begin{definition}
\label{Insert}
Let $A_S$ be a network that contains only trivial cut edges and has a single degree 2 vertex $v$. Let $N$ be a network such that $supp(A_S)\cap supp(N) = \emptyset$, and let $l$ be a leaf on $N$ that is adjacent to internal vertex $u$. Then $I(N, l, A_S)$ is the network that is constructed by removing $l$, adding $A_S$, and replacing the edge $\{u,l\}$ with the edge $\{u,v\}$. This network now contains $A_S$ as an exterior structure with stem $\{u,v\}$. This is the \textit{insertion procedure}.
\end{definition}

We can think of the cut and insertion procedures as inverses of each other, removing or inserting exterior structures, respectively, where $N = I(C(N, A_S), a_S, A_S) $. We can extend the insertion procedure to recursively build up a network given a sequence.  

\begin{definition}
\label{insertseq}
Let $Z$ be any network with at least four leaves and let \\$\{(A_{S_1},a_{S_1}), \dots, (A_{S_k}, a_{S_k})\}$ be a partially ordered set of tuples such that each $A_{S_i}$ is a distinct network that contains only trivial cut edges and has a single degree 2 vertex $v_i$, and each $a_{S_i}$ is a leaf that appears either on $Z$ or on any $A_{S_j}$ where $j<i$. Furthermore, $supp(Z)\cap supp(A_{S_i}) = supp(A_{S_i}) \cap supp(A_{S_j}) = \emptyset \quad  \forall \: i \neq j$. Then, $I(Z \mid (A_{S_1},a_{S_1}), \dots, (A_{S_k}, a_{S_k}))$ is an \emph{insertion sequence} defined by recursively inserting the exterior structures. We let $N_0 = Z$, and $N_i = I(N_{i-1}, a_{S_i}, A_{S_i})$ for $1 \leq i \leq k$. The resulting network $N_k$ is $I(Z \mid (A_{S_1},a_{S_1}), \dots, (A_{S_k}, a_{S_k}))$.
\end{definition}

Theorem~\ref{thm:CProof} demonstrates that every network can be represented by an insertion sequence as shown in Figure~\ref{fig:cutting}.

\begin{figure}
 \centering
 \includegraphics[width=12cm]{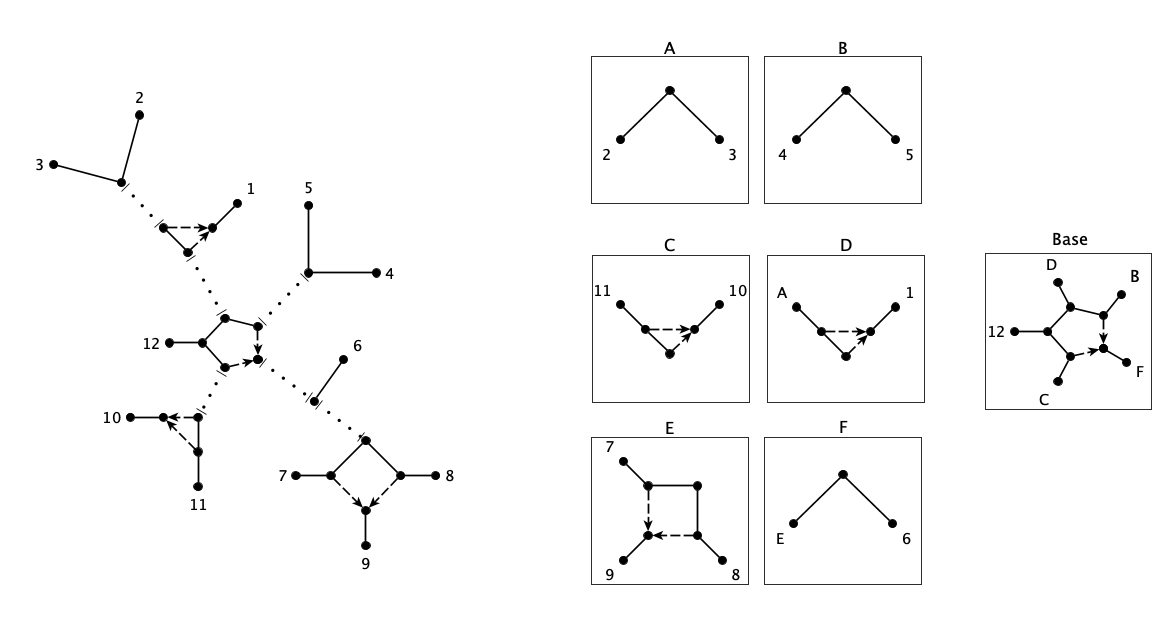}
 \caption{Left: A network with all non-trivial cut edges identified as dotted edges. Right: An ordered set of exterior structures and a base that determines the network on the left.}
 \label{fig:cutting}
\end{figure}

\begin{theorem}
\label{thm:CProof}
Every network $N$ with four or more leaves can be represented by an insertion sequence with a quarnet or a sunlet base.
\end{theorem}

\begin{proof}
If $N$ only has four leaves, then it is quarnet. So we assume $N$ has at least five leaves and induct on the number of non-trivial cut edges in a network. If there are no non-trivial cut edges in $N$, then by definition, the network is a sunlet. Now, assume that every network with at least four leaves and $k$ non-trivial cut edges can be represented by an insertion sequence with a quarnet or sunlet base. Given a network $N$ with $k+1$ non-trivial cut edges, recognize that $N$ has at least five leaves, so there must be an exterior structure $A_S$ such that $X \setminus S$ contains at least three leaves. Define a network $N'$ by replacing $A_S$ with a leaf $l$. Since $N'$ has at least four leaves and $k$ non-trivial cut edges, $N'$ can be represented by an insertion sequence with a quarnet or sunlet base. The claim follows by observing that $N = I(N', l, A_S)$.
\end{proof}

Although we have assumed knowledge of $N$ so far, we need to determine exterior structures from quarnets in practice. The following sections explain how any exterior structure can be identified from quarnets. We begin with the identifiability of the exterior blob structure. 

\subsection{Recovering The Topology of Exterior Blobs}

We rely on properties of sunlets to guide the identification of exterior blob topology. Recall that an $n$-sunlet consists of $n$ leaves and one semi-directed cycle of length $n$. The symmetry of a sunlet simplifies the analysis of its displayed quarnets (see Figure~\ref{fig:SunletRestriction}).

\begin{figure}
 \centering
 \includegraphics[width=8cm]{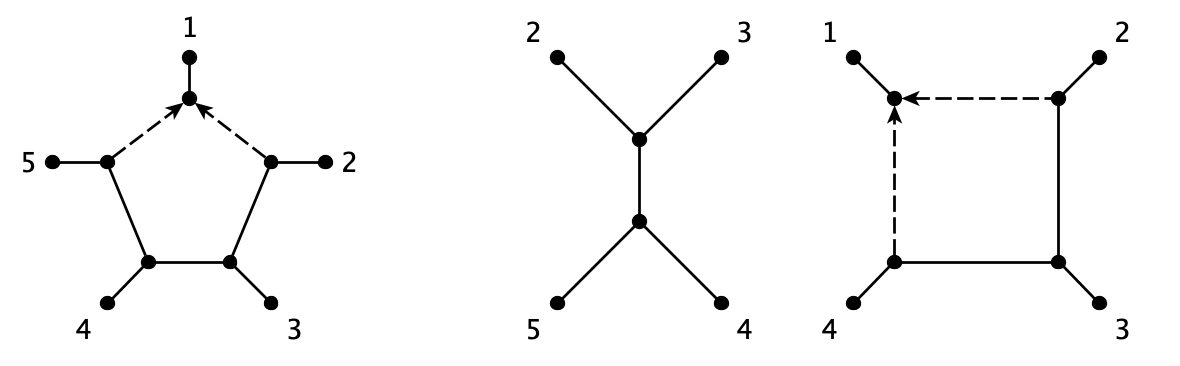}
 \caption{Left: A 5-sunlet network $N$. Right: A quartet tree and a square quarnet displayed by $N$.}
 \label{fig:SunletRestriction}
\end{figure}

\begin{lemma}
\label{lem:outleafOrcherry}
Every network with at least four leaves contains an out leaf or a tree cherry.
\end{lemma}

\begin{proof}
Let $N$ be a network with at least four leaves. Root $N$ at a valid location $\rho$ and direct all edges away from the root. Internal vertices either have in-degree one and out-degree two or, in the case of reticulation vertices, in-degree two and out-degree one. 

Choose a leaf $l$ that maximizes the path length from $\rho$. Since there are at least four leaves, $l$ must be adjacent to a non-root vertex $v$. If $v$ is a reticulation vertex, $l$ is an out leaf due to the orientation of the edges. Otherwise, $v$ must have out-degree two, implying that $v$ is part of a tree cherry.
\end{proof}

\begin{theorem}
\label{thm:sunqnets}
A network $N$ is an $n$-sunlet with $n \geq 4$ if and only if $\mathcal{Q}(N)$ consists of $\binom{n-1}{3}$ squares and $\binom{n-1}{4}$ quartet trees.
\end{theorem}

\begin{proof}
Let $N$ be an $n$-sunlet. By definition, it has $n$ leaves and an $n$-length cycle with exactly one reticulation event. According to Lemma~\ref{lem:outleafOrcherry}, one of the leaves must be an out leaf since the sunlet does not have any tree cherries. Consequently, any quarnet containing this out leaf is a square, while any quarnet not containing an out leaf forms a quartet tree (see Figure~\ref{fig:SunletRestriction}). Thus, there are $\binom{n-1}{3}$ quarnets with the out leaf (squares) and $\binom{n-1}{4}$ quarnets without the out leaf (quartet trees). 

Now, assume that $\mathcal{Q}(N)$ consists of $\binom{n-1}{3}$ squares and $\binom{n-1}{4}$ quartet trees for some network $N$. Since at least one reticulation vertex is present in the quarnets, $N$ must contain at least one reticulation vertex, implying that it is not a tree. Suppose for contradiction that $N$ is not a sunlet; then, $N$ must either contain a tree cherry or another cycle. 

If $N$ has a tree cherry, it follows from the restriction process that at least one single triangle contains both the reticulation vertex and the tree cherry in the complete quarnet set. Conversely, if $N$ has another cycle, the restriction process will yield at least one double triangle containing two reticulation vertices. In either scenario, $\mathcal{Q}(N)$ includes quarnets that are neither square nor quartet trees, leading to a contradiction. Therefore, $N$ must be a sunlet.
\end{proof}

Thus, we can identify whether a network is a sunlet based on the count of squares and trees in $\mathcal{Q}(N)$. Moreover, we can determine the reticulation's location and the leaves' ordering along the sunlet based on $\mathcal{Q}(N)$. Finding this ordering is similar to the sequential algorithm: we start with a single square quarnet and add one leaf at a time based on information from other quarnets. This process is detailed in Algorithm~\ref{alg:SunConstruct}.

\begin{algorithm}[H]
\caption{Sunlet Construction}
\label{alg:SunConstruct}
\DontPrintSemicolon
\KwInput{Let $\mathcal{Q}_{sq(N)} \subset \mathcal{Q}(N)$ be the set of square reticulation quarnets in the complete quarnet set of a sunlet network $N$.}
\KwOutput{The placement of the reticulation vertex and the ordering of the leaves of $N$, denoted as $Sunlet(\mathcal{Q}_{sq(N)}))$.}

Pick $q \in \mathcal{Q}_{sq(N)}$. Set $B_4 = q$. 

\For{$i \gets 4$ \KwTo $n-1$}{
 Pick a quarnet $q_{i+1} \in  \mathcal{Q}_{sq(N)}$ such that $|supp(q_{i+1}) \cap supp(B_i)| = 3$, and two of the common leaves, $x_a$ and $x_b$, form a pendant pair in  $B_i$ but not in $q_{i+1}$.

 Add a new vertex $u$.
 
 Replace the edge $e$ with the pair of edges $\{v,u\}$ and $\{u,w\}$.
 
 Let $x_{i+1}$ be the element in $supp(q_{i+1}) \setminus supp(B_i)$.
 
 Add an edge $\{u,x_{i+1}\}$. 

 Call this new network $B_{i+1}$.
}
\Return{$B_n$}
\end{algorithm}

\begin{theorem}
\label{thm:sunlet}
Suppose $N$ is a sunlet network. Let $ \mathcal{Q}_{sq(N)} \subset \mathcal{Q}(N) $ be the set of square quarnets in the complete quarnet set of $ N $. Then $Sunlet(\mathcal{Q}_{sq(N)}) = N $.
\end{theorem}

\begin{proof}
Let $ \mathcal{Q}_{sq(N)} \subset \mathcal{Q}(N) $ be the set of square quarnets in the complete quarnet set of a sunlet network $ N $. We will prove that $ Sunlet(\mathcal{Q}_{sq(N)}) = N $ by induction on the number of leaves $ n $ of the sunlet.

For the base case, when $ n = 4 $, we have $ \mathcal{Q}_{sq(N)} = N $, so $ Sunlet(\mathcal{Q}_{sq(N)}) $ returns $ N $ in the initial step. Now assume $Sunlet(\mathcal{Q}_{sq(N)}) = N $ for all sunlets with $ n $ or fewer leaves, and consider $ N $ as an $ (n+1) $-leaf sunlet.

Notice that Algorithm~\ref{alg:SunConstruct} induces an ordering of the leaves of $ N $, denoted as $ \{x_1, x_2, \ldots, x_{n+1}\} $, where $ \{x_1, x_2, x_3, x_4\} $ are the support of the initial quarnet, and $ x_i = supp(q_i) \setminus supp(B_i) $ for $ 5 \leq i \leq n+1 $. Therefore, the steps for $ i = 1 $ to $ i = n $ of the Sunlet Construction Algorithm are equivalent to running the Sunlet Construction Algorithm on the restriction of $ N $ to the first $ n $ leaves.

By the inductive hypothesis, $ B_n $ is the restriction of $ N $ to the first $ n $ leaves. Since $ \mathcal{Q}_{sq(N)} $ consists of the square quarnets of a sunlet network, there exists a unique choice of square reticulation quarnet $ q_{n+1} $ such that $ |supp(q_{n+1}) \cap supp(B_n)| = 3 $, where two of the common leaves, $ x_a $ and $ x_b $, are pendant in $ B_n $ and not pendant in $ q_{n+1} $. Given this, the remaining portion of the sunlet construction places $ x_{n+1} $ in its proper location within the sunlet. Thus, we conclude that $ Sunlet(\mathcal{Q}_{sq(N)}) = N $.
\end{proof}

We can extend the utility of Algorithm ~\ref{alg:SunConstruct} by recognizing that for any exterior blob with support $S$, where $|S|\ge 3$, the restriction to a quarnet containing any 3 leaves in $S$ and any 1 leaf not in $S$ will be a sunlet.

\begin{corollary}
\label{cor:extblob}
Let $A_S$ be an exterior blob of a level-1 network $N$. If $S$ is the support of the exterior blob of $A$, and $y$ is any other leaf of $N$, then \begin{equation*} Sunlet(\mathcal{Q}(N|_{S \cup \{y\}})) = N|_{S \cup \{y\}}. \end{equation*}
\end{corollary}

\begin{proof}
Let $A$ be an exterior blob of a level-1 network $N$. If $S$ is the set of leaves of the exterior blob of $A$, and $y$ is any other leaf of $N$, then the restriction $ N|_{S \cup \{y\}} $ is a sunlet. Thus, by Theorem~\ref{thm:sunlet}, $ Sunlet(\mathcal{Q}_{sq(N|_{S \cup \{y\}})}) = A $.
\end{proof}

Thus, the exterior blobs of the network $N$ can be reconstructed from the subset of $ \mathcal{Q}(N) $ whose support is contained in $ S \cup \{y\} $ using the Sunlet Construction Algorithm. 

\subsection{Efficient Identification of Exterior Structure Leaf Sets}
The application of the Sunlet Construction Alogithm requires knowing in advance which subsets of leaves form the exterior blobs or checking all $ 2^n $ possible subsets of leaves. We, therefore, need a strategy for determining which subsets of the leaves are contained in an exterior structure. Once this has been identified, one can reconstruct the exterior structure using Corollary~\ref{cor:extblob}. Our goal is to identify a collection of subsets of leaves that could support exterior blobs. To achieve this, we introduce the incompatibility graph and two lemmas that clarify the use of the incompatibility graph in the identification of exterior blobs.

\begin{definition}
\label{Incompatgraph}
For a network $N$ with leaf set $X$, let $H(\mathcal{Q}(N))$ be the \emph{incompatibility graph} with vertex set  $V = X $. Define an edge between leaves $a$ and $b$ if there exists a single or double triangle $q \in \mathcal{Q}(N)$ such that $ \{a,b\} \subset supp(q)$ and there is a non-trivial cut edge on the path between $ a$ and $b$ on $q$ (see Figure~\ref{fig:IncomSingleQ}).
\end{definition}

\begin{figure}[ht]
 \centering
 \includegraphics[width=6cm]{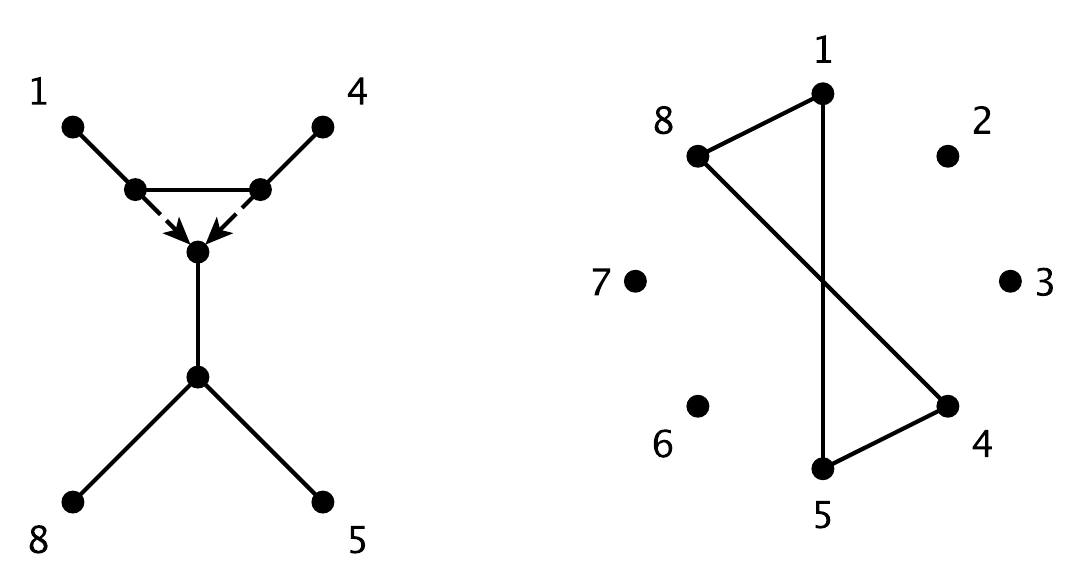}
 \caption{Left: A single triangle $q \in \mathcal{Q}(N)$. Right: the four edges of the incompatibility graph   on the leaves of $N$ induced by $q$.}
 \label{fig:IncomSingleQ}
\end{figure}

\begin{lemma}
\label{lemma:extBlob}
Let $N$ be a level-1 network with at least four leaves and $S$ be the support of an exterior blob $A_S$. Then for any $x \in S$, $\{x,y\}$ is an edge in $H(\mathcal{Q}(N))$ if and only if $y\notin S$.
\end{lemma}

\begin{proof}
Let $N$ be a level-1 network, and let $A_S$ be an exterior blob of $N$ with support $S$. 

Let $x,y$ be any two distinct leaves in the support of $A_S$. Suppose $ q $ is a quarnet whose support contains $x,y \in S$. If at least one of the other elements of $supp(q)$ is also contained in $S$, then $ q $ can be viewed as a quarnet on the restriction of $ N $ to the sunlet containing $ A_S $ and the additional taxa in the support of $ q $. Thus, by Theorem~\ref{thm:sunqnets}, $ q $ must either be a tree or a square, which would not generate an edge between $x$ and $y$ in $ H(\mathcal{Q}(N)) $. 
If, instead, the support of $ q $ contains two elements outside of $ S $, then the quarnet must be a quartet or a single or triangle with $x$ and $y$ on the same side of the cut edge that defines the exterior blob. Thus, no edge is drawn between $x$ and $y$. There are no edges connecting pairs of vertices in $ S $.

Conversely, let $x \in S$ and $y \notin S$. If the stem tip of $A_S$ is the reticulation vertex, then any quarnet with exactly 2 elements in $S$ will have a reticulation cherry between those two elements, so choose $a \in S$ arbitrarily. If instead the reticulation vertex is not the stem tip of $A_S$, then $A_S$ must have an out leaf. Any quarnet containing the out leaf and exactly one other element of $S$ will have a reticulation cherry between those two elements, so choose $a \in S$ such that either $a$ or $x$ is the out leaf. Then for any $b \notin S$, the quarnet $q$ with $x,y,a,b \in supp(q)$ will be a single or double triangle with $x$ and $a$ on a reticulation cherry and $y$ on the other side of the cut edge. Thus $\{x,y\}$ is an edge in $ H(\mathcal{Q}(N)) $.

\end{proof}

A proper graph coloring is a partitioning of vertices into color classes such that no two vertices in a single color class are adjacent \cite{GraphColoring}. By construction, the incompatibility graph created from the full quarnet set of a network does not contain edges between the vertices in the support of a single exterior blob. We show that it is possible to find a coloring that places every vertex in the support of an exterior blob into the same color class. 

\begin{theorem}
\label{thm:extBlob}
Let $N$ be a level-1 network and $A = \{x_1, x_2, \ldots, x_k\}$ be the support of an exterior blob of $N$. Then $A$ must be a color class of the vertex coloring of the incompatibility graph $H(\mathcal{Q}(N))$ under the greedy coloring algorithm.
\end{theorem}

\begin{proof}
Let $N $ be a level-1 network and $ A = \{x_1, x_2, \ldots, x_k\} $ be the support of an exterior blob of $ N $.  Let $(a_1,a_2,a_3,\cdots, a_n)$ be any ordering of the vertices of $N$. We can re-index the elements of $A$ such that $x_1$ through $x_n$ form an increasing subsequence of the ordering.

Then using the greedy coloring algorithm $x_1$ would be placed in its own color class, as by Lemma~\ref{lemma:extBlob}, $x_1$ would be adjacent to all of the preceding vertices in $H(\mathcal{Q}(N))$. Call this color class $C$. Since $x_1$ is adjacent to all elements outside of $A$ no element which is not in the support of the exterior blob can be placed in $C$. Now note that all additional elements in $A$ should be placed in this color class, as they would be adjacent to any of the elements not in $A$ (and thus not be able to be placed in any existing coloring class except for $C$) but would not be adjacent to any element in the color class $C$ (by Lemma \ref{lemma:extBlob}). So it follows that $C=A$.
\end{proof}

To determine if the leaves in a color class $ X' $ form an exterior blob, one can pick a leaf $ y $ that is not in the color class and check if the quarnets that have support in $ X' \cup \{y\} $ are associated with a sunlet network. Thus, the graph $ H(\mathcal{Q}(N)) $ allows us to check a maximum of $ n $ sets for exterior structures rather than the potential $ 2^n $ subsets.

Unlike exterior blobs, cherries can be identified simply by counting quarnets.

\begin{theorem}
\label{idCherries}
A network $ N $ has a cherry containing leaves $ a $ and $ b $ if and only if $ \mathcal{Q}(N) $ contains $ \binom{n-2}{2} $ quarnets with a cherry containing $ a $ and $ b $.
\end{theorem}

\begin{proof}
If $ N $ has a cherry containing leaves $ a $ and $ b $, then it follows from the restriction process that $ a $ and $ b $ must be contained on a cherry in any quarnet they appear together. Therefore, there are $ \binom{n-2}{2} $ such quarnets.

Now assume that $ \binom{n-2}{2} $ quarnets have a cherry containing leaves $ a $ and $ b $. In this case, $ a $ and $ b $ are contained on a cherry in every quarnet where they appear together. This implies that no path in the network between $ a $ and $ b $ contains a non-trivial cut edge. Furthermore, any path between $ a $ and $ b $ and two other leaves must contain a non-trivial cut edge since there exists a quarnet with a non-trivial cut edge connecting $ a $ and $ b $ to the other leaves. Therefore, a cut edge in the network exists that partitions the leaf set into $ \{a,b\} $ and the remaining leaves, confirming that $ a $ and $ b $ are contained in a cherry.
\end{proof}

Consequently, all exterior structures of $ N $ can be identified from the complete quarnet set.

\subsection{Cherry-Blob Algorithm}
\label{cbAlg}

Now that we can identify exterior structures from quarnets, we need a method to order these structures so that an insertion sequence can be performed, attaching all structures to form $ N $.   In order to iterate the process of identifying exterior structures of a network, we introduce the quarnet refinement process in Algorithm~\ref{alg:CutAlg} to represent the quarnet set for the network resulting from cutting $ A_{S} $ from $ N $

\begin{algorithm}[]
\label{alg:CutAlg}
\DontPrintSemicolon
\KwInput{A complete set of quarnets $ \mathcal{Q}(N) $ and and exterior structure $A_S$ of $N$}
\KwOutput{A complete set of quarnets $ \mathcal{Q}(C(N,A_S)) $}

 Set $ \mathcal{Q}(C(N,A_S)) = \emptyset $\\
 \For{$q \in \mathcal{Q}(N)$}{
    \uIf{$ supp(q) \cap S = \emptyset $}{
        $ \mathcal{Q}(C(N,A_S)) = \mathcal{Q}(C(N,A_S)) \cup \{q\} $
    }
    \uElseIf{$ supp(q) \cap S = \{s\} $}{
        Define $ q' $ by taking the quarnet $ q $ and replacing $ s $ by $ a_S $\\
        $ \mathcal{Q}(C(N,A_S)) = \mathcal{Q}(C(N,A_S)) \cup \{q'\} $
    }
}
\Return{ $ \mathcal{Q}(C(N,A_S)) $}
\caption{Quarnet Refinement Procedure}
\end{algorithm}

The quarnet refinement procedure enables us to iterate through the process of cutting exterior structures one at a time until we are left with either a single quarnet or the quarnet set for a sunlet. A \emph{terminal quarnet set} is defined as a complete quarnet set that contains either a single quarnet or only square reticulations and quartet trees. We keep track of the exterior structures cut along the way, allowing us to construct an insertion sequence that begins with the base to build the parent network. Algorithm~\ref{alg:CBlob} formalizes this iterative process.

\begin{algorithm}[H]
\label{alg:CBlob}
\DontPrintSemicolon
\KwInput{A complete set of quarnets $ \mathcal{Q}(N) $}
\KwOutput{The network $ CherryBlob(\mathcal{Q}(N)) = N $}

Set $ \mathcal{Q} = \mathcal{Q}(N) $

Set $ i = 0 $

Set $T = \emptyset$

\While{ $ \mathcal{Q} $ is not a terminal quarnet set}{
    \While{there exists a cherry in $ \mathcal{Q} $ (Theorem~\ref{idCherries})}{
        Set $ i = i + 1 $\\
        Define the placeholder $ a_{S_i} $ for the cherry $ A_{S_i} $\\
        Append $T = [(A_{S_i}, a_{S_i}), T]$

        Set $ \mathcal{Q} = \mathcal{Q}(C(N,A_{S_i})) $ using Algorithm~\ref{alg:CutAlg}
    }
    Compute $ H(\mathcal{Q}) $

    Use color classes of $ H(\mathcal{Q}) $ to determine if there is an exterior blob

    \If{there exists an exterior blob}{
        Set $ i = i + 1 $

        Construct $A_{S_i}$ using Algorithm ~\ref{alg:SunConstruct}

        Define the placeholder $ a_{S_i} $ for a single exterior blob $ A_{S_i} $\\
        Append $T = [(A_{S_i}, a_{S_i}), T]$

        Set $ \mathcal{Q} = \mathcal{Q}(C(N,A_{S_i})) $ using Algorithm ~\ref{alg:CutAlg}
    }
}
\eIf{ $ \mathcal{Q} $ is a single quarnet }{
    Call the quarnet $ Z $
}{
    Use Algorithm~\ref{alg:SunConstruct} to find the sunlet defined by $ \mathcal{Q} $ and call it $ Z $
}

Perform the insertion sequence $ I(Z \mid (A_{S_k},a_{S_k}), \dots, (A_{S_1}, a_{S_1})) $ using base $ Z $, exterior structures $ \{A_{S_k}, \ldots, A_{S_1}\} $, and leaf set $ \{a_{S_k}, \ldots, a_{S_1}\} $ such that each $ A_{S_i} $ replaces leaf $ a_{S_i}$

\Return{$N= I(Z \mid T)$}

\caption{Cherry-Blob Algorithm}
\end{algorithm}

\begin{theorem}
Given a complete set of quarnets $ \mathcal{Q}(N) $ of a level-1 network $ N $, we have $ CherryBlob(\mathcal{Q}(N)) = N $.
\end{theorem}

\begin{proof}
Let $ \mathcal{Q}(N) $ be the complete set of quarnets of a level-1 network $ N $. The proof that $ CherryBlob(\mathcal{Q}(N)) = N $ uses induction on the number $ k $ of non-trivial cut edges of the network. If $ k = 0 $, then $ N $ is a sunlet. Since there are no cherries or exterior blobs in $ N $, we have $ CherryBlob(\mathcal{Q}(N)) = Sunlet(\mathcal{Q}(N)) $, and by Theorem~\ref{thm:sunlet}, $ Sunlet(\mathcal{Q}(N)) = N $.

Now assume $ N $ has $ k+1 $ cut edges. Then $ N $ must contain either a cherry or an exterior blob. If $ N $ contains a cherry, the leaves of the cherry will be identified (using Theorem~\ref{idCherries}) in the first sub-while loop of the Cherry Blob Algorithm, and the location of the placeholder leaf will be preserved when performing Algorithm~\ref{alg:CutAlg}.

If there is not a cherry but an exterior blob, the leaves will be identified in the if statement within the while loop (using Theorem~\ref{thm:extBlob}). In this case, referencing the incompatibility graph reduces the number of cases that need to be checked as potential exterior structures. The orientation of the blob will be identified using algorithm ~\ref{alg:SunConstruct} in the same step, and the placeholder leaf will be preserved when performing algorithm ~\ref{alg:CutAlg}.

In both cases, we have identified $ N = I(N', a_k, A_k) $ for the identified exterior structure $ A_k $. We note here that $ N' $ is a network with one fewer cut edge than $ N $; therefore, by our inductive hypothesis, $ CherryBlob(\mathcal{Q}(N')) = N' $. The claim follows by noting that $ \mathcal{Q}(N') $ is the quarnet refinement of $ (\mathcal{Q}, A_k) $.
\end{proof}

The following theorem demonstrates that the cherry-blob algorithm runs in polynomial time.
\begin{theorem}
The \textsc{Cherry-Blob Algorithm}, when given as input a complete set of quarnets $\mathcal{Q}(N)$ on $n$ leaves, runs in time $O(n^6)$ in the worst case.
\end{theorem}

\begin{proof}
We analyze the runtime of the algorithm by considering each phase of its execution.

\textbf{Initial Setup.}
The algorithm begins by setting $\mathcal{Q} = \mathcal{Q}(N)$, which we assume is provided as part of the input. Accessing or copying this data takes $O(n^4)$ time.

\textbf{Main While Loop.}
The algorithm proceeds while $\mathcal{Q}$ is not a terminal quarnet set. Checking if  $\mathcal{Q}$ is terminal is $O(n^4)$ as it checks if each quarnet is a square. Each iteration either reduces the number of taxa by removing a cherry or reduces an exterior blob, both of which strictly decrease the size of the network. Thus, the outer loop iterates at most $n$ times.

We examine the two primary operations that can occur in each iteration:

\textbf{Cherry Reductions.}
Within the inner \texttt{while} loop, the algorithm identifies a cherry and updates the quarnet set to reflect its contraction. Identifying cherries requires examining all $O(n^4)$ quarnets and checking whether any pair of taxa appear in $\binom{n-2}{2} = O(n^2)$ such quarnets, so this step takes $O(n^4)$ time. 

Updating $\mathcal{Q}$ after a cherry contraction involves recomputing or filtering the quarnets and takes $O(n^4)$ time. Hence, each cherry reduction iteration takes $O(n^4)$ time.

Since there can be at most $n$ cherry reductions (each removing one taxon), the total cost of all cherry reductions is $O(n \cdot n^4) = O(n^5)$.

\textbf{Exterior Blob Reductions.}
If no cherry exists, the algorithm constructs a graph $H(\mathcal{Q})$ on $n$ vertices, where each quarnet may induce up to 4 edges. Hence, constructing $H(\mathcal{Q})$ requires $O(n^4)$ time.

Greedy coloring of this graph can be performed in $O(v + e) = O(n + n^2) = O(n^2)$ time. For each color class (at most $n$), we check whether it supports an exterior blob. This involves checking all quarnets that include 3 taxa from the color class and 1 from outside, totaling $O(n^4)$ checks per color class. Thus, all color classes are checked in $O(n \cdot n^4) = O(n^5)$ time.

If an exterior blob is found, Algorithm~\ref{alg:CutAlg} is used to refine the quarnet set, which again takes $O(n^4)$ time.

Hence, each blob reduction iteration takes $O(n^5)$ time in the worst case. Since the main loop runs at most $n$ times, the total cost of all blob reductions is $O(n \cdot n^5) = O(n^6)$.

\textbf{Sunlet Construction.}
Once the main loop terminates, either $\mathcal{Q}$ is a single quarnet (trivial case) or a sunlet network is constructed using Algorithm~\ref{alg:SunConstruct}. Each of the $O(n)$ steps in sunlet construction involves scanning at most $O(n^4)$ quarnets and checking if pairs of leaves are pendant on a sunlet which can be done in $O(n)$ steps since the number of leaves in a sunlet network are a linear function of the number of edges. The appropriate insertion of the new leaf into the sunlet network is completed in a constant number of steps. Thus the total cost is $n \times O(n^4 \cdot n) = O(n^6)$.

\textbf{Insertion Sequence.}
Finally, the algorithm performs an insertion sequence on the base sunlet $Z$ using the exterior structures. If the network is stored with adjacency lists or matrices requiring $O(n^2)$ space, each of the $O(n)$ insertions takes $O(n^2)$ time, giving a total of $O(n \cdot n^2) = O(n^3)$ for this step.

\textbf{Conclusion.}
The total runtime of the algorithm is:
\begin{align*}
O(n^4) \text{ (setup)} + O(n^5) &\text{ (check terminal)} + O(n^5) \text{ (cherry reductions)} \\ + \ O(n^6) \text{ (blob reductions) } + &O(n^6) \text{ (sunlet construction) } + O(n^3) \text{ (insertions) }\\ &= O(n^6).
\end{align*}

\end{proof}


\section*{Competing interests}
The authors declare that they have no competing interests.

\section*{Author's contributions}
All authors helped design and implement this project. All authors helped draft the manuscript and have read and approved the final manuscript.
\section*{Acknowledgements}
 The authors wish to thank Yifei Tao who made significant contributions to an earlier version of this manuscript. This material is based upon work supported by the National Science Foundation under Grant No. DMS-1757616 and Grant No. DMS-1616186. J.R. was supported by the National Science Foundation under Grant No. DMS-1929284 while in residence at the Institute for Computational and Experimental Research in Mathematics in Providence, RI, during the Theory Methods and Applications of Quantitative Phylogenomics program.

\bibliographystyle{plain} 
\bibliography{references.bib}      

\end{document}